\documentclass[conference]{IEEEtran}
\IEEEoverridecommandlockouts

\usepackage{cite}
\usepackage{amsmath,amssymb,amsfonts}
\usepackage{algorithmic}
\usepackage{graphicx}
\usepackage{textcomp}
\usepackage{comment}
\usepackage{xcolor}
\usepackage{braket} 
\usepackage{mathtools} 
\usepackage[section]{placeins} 
\usepackage{amsthm}
\newtheorem{theorem}{Theorem}
\newtheorem{proposition}{Proposition}
\newtheorem{corollary}{Corollary}
\usepackage{tikz}
\usetikzlibrary{quantikz}

\newcommand{\om}{\omega}
\newcommand{\Zd}{\mathbb{Z}_d}
\newcommand{\Id}{\mathbb{I}_d}
\newcommand{\Cd}{\mathcal{C}_d}

\def\BibTeX{{\rm B\kern-.05em{\sc i\kern-.025em b}\kern-.08em
    T\kern-.1667em\lower.7ex\hbox{E}\kern-.125emX}}

\begin{document}

\title{Generalizing Pauli Checks for Qudit-based \\ Quantum Error Detection and Mitigation\\
\thanks{This work was supported by the U.S. Department of Energy, Office of Science, National Quantum Information Science Research Centers, Superconducting Quantum Materials and Systems Center (SQMS), under Contract No. 89243024CSC000002. Fermilab is managed by FermiForward Discovery Group, LCC, acting under Contract No. 89243024CSC000002.\\ \\
$^*$These authors contributed equally to this work.\\ \\
This paper has been accepted to the 2026 IEEE International
Conference on Quantum Computing and Engineering (QCE),
IEEE Quantum Week 2026.}
}
\author{\IEEEauthorblockN{Noble Agyeman-Bobie$^*$}
\IEEEauthorblockA{\textit{Department of Mathematics and Physics} \\
\textit{Grambling State University}\\
Grambling, LA, USA \\
nagyeman@gsumail.gram.edu}
\and
\IEEEauthorblockN{Quinn Langfitt$^*$}
\IEEEauthorblockA{\textit{Department of Computer Science} \\
\textit{Northwestern University}\\
Evanston, IL, USA \\
qlangfitt@u.northwestern.edu}
\and
\IEEEauthorblockN{Salahedeen Issa}
\IEEEauthorblockA{\textit{Department of Computer Science} \\
\textit{Princeton University}\\
Princeton, NJ, USA \\
salahedeen@princeton.edu}
\and[\hfill\mbox{}\par\mbox{}\hfill]
\IEEEauthorblockN{Nikos Hardavellas}
\IEEEauthorblockA{\textit{Department of Computer Science} \\
\textit{Northwestern University}\\
Evanston, IL, USA \\
nikos@northwestern.edu}
\and
\IEEEauthorblockN{Kaitlin N. Smith}
\IEEEauthorblockA{\textit{Department of Computer Science} \\
\textit{Northwestern University}\\
Evanston, IL, USA \\
kns@northwestern.edu}
}

\maketitle

\begin{abstract}

  Pauli Check Sandwiching (PCS) is a quantum error detection (QED) technique that protects a quantum circuit by utilizing a pair of controlled Pauli operators, or checks, and detecting errors that anti-commute with the checks. Further, PCS can be used for quantum error mitigation (QEM) via post-selection based on the Pauli check syndrome values. Currently, PCS is leveraged in the qubit space. In this paper, we introduce a generalized approach for applying PCS-based QED and QEM to quantum information of arbitrary dimension in the Hilbert space. Each pair of these extended checks consists of a sequence of gates in the Heisenberg-Weyl operator set that extend Pauli operators into the qudit space. These qudit checks use at least one ancilla qudit to detect qudit errors that do not commute with the unitary selected for the check. We show that our proposed methods for qudit QED can detect errors of arbitrary dimensions. More specifically, we prove that an arbitrary Heisenberg-Weyl error maps deterministically to a unique ancilla readout, and further, post-selecting on the $|0 \rangle$ readout guarantees unit fidelity in the noiseless check limit. We validate these findings numerically across dimensions $d=2$ through $d=9$, achieving error-mitigated fidelities above $97.5\% $ under realistic depolarizing error rates. 

\end{abstract}


\section{Introduction}

The field of quantum computing is overwhelmingly focused on leveraging ``quantum bits,'' or qubits, that are two-dimensional. When considering a single qubit, information can be either in one of two basis states, or a superposition of them both. However, many quantum technologies, such as trapped ions~\cite{low2020practical}, superconducting circuits~\cite{fischer2023universal}, and neutral atoms~\cite{lindon2023complete}, contain additional energy levels that can be leveraged for higher-dimensional information encoding. This enables a ``quantum digit,'' or a qudit. Because of the potential to efficiently explore high-dimensional spaces with these types of systems, there is a growing interest in developing methods for encoding qudits in a scalable manner. For example, architectures based on superconducting radio-frequency (SRF) cavities are emerging~\cite{reineri2023exploration}, and high-dimensional state preparation with  $d=20$ has recently been demonstrated~\cite{kim2025ultracoherent}.  

Developments in physical qudit realization motivate the exploration of how to unlock the full promise of these systems. Just as with qubit-based computation, high-dimensional entanglement can be prepared \cite{icrc19}, \cite{entang_journal}, and the work in Ref. \cite{baker2020efficient} demonstrates that the ability to compress many qubit states into a single qudit can free up ancillas for qubit-based computation. The benefits of leveraging qudits during computation include efficient circuit synthesis and reduction in resource requirements, among others, which could potentially make these types of quantum computers well-suited for high-dimensional applications such as lattice gauge theory and high energy physics \cite{kurkccuoglu2024qudit}. When qudits are used for storing intermediate states, it has been shown that temporary qudit levels within circuit decompositions enable significant improvements in circuit depth \cite{gokhale2019asymptotic}.

In lockstep with advances made on the theoretical front, qudit hardware progress has led to various proposals of universal gate sets for qudit computation on physical qudits~\cite{job2023efficient, kim2025ultracoherent}. However, much work remains to discover the best protocols that help  manage noise in qudit-based quantum systems during runtime. While increasing system dimension helps increase information capacity, it also increases the amount and complexity of the types of errors that can occur. As an example, a simple basis flip channel that could occur in a qubit scheme generalizes to noise that exchanges probability amplitudes between the $d$ orthonormal vectors that span the dimension-$d$ Hilbert space of a qudit. Because of the challenge that noise presents, effective error management techniques are necessary if qudit-based quantum systems are to become practical. 

In our work, we develop new techniques for qudit quantum error detection (QED) and quantum error mitigation (QEM) inspired by qubit-based Pauli-check sandwiching (PCS).
In particular, we construct generalized left and right Pauli checks for arbitrary dimension $d$ using the Heisenberg-Weyl operator set. We then prove that an arbitrary Heisenberg-Weyl error acting on a data qudit maps deterministically to a unique ancilla readout (Propositions~1 and~2), and that joint post-selection on both check layers yields unit fidelity in the noiseless-check limit (Theorem~1). Finally, we validate the protocol numerically across dimensions $d = 2$ through $d = 9$ on a single-gate modular addition circuit (Mod-2) and a two-qudit generalized GHZ state preparation circuit (GHZ-2), demonstrating error-mitigated fidelities above 97.5\% at realistic depolarizing error rates.

\section{Background}
\subsection{Higher-dimensional Quantum Information}

In general, the form of a single qudit for an arbitrary dimension, $d$, is given by

\begin{equation}\label{eq:gen_qudit}
\Ket{\psi_d} = \sum_{i=0}^{d-1} \alpha_i \Ket{i_d}.
\end{equation}

\noindent The amplitudes $\alpha_i$ in Eq.~\ref{eq:gen_qudit}
are complex-valued quantities that satisfy the Born rule such that

\begin{equation}\label{eq:born-rule-qudit}
\sum_{i=0}^{d-1} |\alpha_i|^2=1.
\end{equation}

\noindent Just as in qubit-based systems, a quantum system of dimension-$d$ qudits can also exhibit superposition, and this superposition is leveraged for quantum parallelism that enables processing of multiple valuations of information in a single quantum computation. High-dimensional maximal superposition for a single qudit can be achieved using the Chrestenson gate~\cite{zilic2007scaling}. The elements of the Chrestenson matrix are powers of the $d^{th}$ roots of unity, $\omega = e^{2\pi i/ d}$ \cite{b6, zilic2007scaling}. 


%

Each element of the dimension-$d$ Chrestenson matrix takes the form $\omega^{jk}$, where $j$ is determined by the column index and $k$ is determined by the row index. In this indexing scheme, the indices $j$ and $k$ begin with $j=k=0$ and increase to $j=k=(d-1)$. It is observed that, for the case $d=2$, the Hadamard matrix for maximal qubit superposition results. Thus, the Chrestenson transform matrices can be considered as generalizations of the Hadamard transform for higher-dimensioned systems. The generalized Chrestenson transform matrix, $\mathbf{C}_d$, is


\begin{equation} \label{eq:gen_Chrestenson_1}
\renewcommand\arraystretch{1}
\mathbf{C}_d =
\frac{1}{\sqrt{d}}\begin{bmatrix*}[c]
    \omega^{0 \cdot 0} & \omega^{1 \cdot 0} & \dots & \omega^{(d-1) \cdot 0} \\
    \omega^{0 \cdot 1} & \omega^{1 \cdot 1} &  \dots & \omega^{(d-1) \cdot 1} \\
    \vdots & \vdots & \ddots & \vdots  \\
    \omega^{0 \cdot (d-1)} & \omega^{1 \cdot (d-1)}  & \dots & \omega^{(d-1) \cdot (d-1)}
\end{bmatrix*}.
\end{equation}

\noindent Using Eq.~\ref{eq:gen_Chrestenson_1}, the $d=4$ Chrestenson gate, $\mathbf{C}_4$, is found to be

\begin{equation} \label{eq:Chrestenson_4}
\renewcommand\arraystretch{1}
\mathbf{C}_4 =
\frac{1}{2}\begin{bmatrix*}[c]
    1 & 1 & 1 & 1 \\
    1 & i &  -1 & -i \\
    1 & -1 & 1 & -1  \\
    1 & -i  & -1 & i   
\end{bmatrix*}.
\end{equation}

\noindent An example physical implementation of the $\mathbf{C}_4$ qudit operator can be found in Ref.~\cite{smith2018radix}.

The Pauli operators from qubit-based quantum computation can also be generalized into higher dimensions. Mathematically, the Pauli-$\mathbf{X}$ operator can be considered as a modulo-2 addition-by-one operation since it transforms qubit $\Ket{0_2}$ into $\Ket{((0+1) \text{mod} \; 2)_2}=\Ket{1_2}$ and $\Ket{1_2}$ into $\Ket{((1+1) \text{mod} \; 2)_2}=\Ket{0_2}$. Thus, the Pauli-$\mathbf{X}$  gate can be understood as a basis permutation operation due to modulo-$k$ addition with respect to modulus $d=2$. The single qudit modulo-addition operations that we will use in this work are denoted as $\mathbf{M}^n_d$. For qudit systems of dimension-$d$, there are $d-1$ distinct non-trivial single-qudit $\mathbf{M}^n_d$ operators. For example, a $d=4$ quantum system would have a total of three non-trivial modulo-addition operations that permute basis vectors. The non-trivial $\mathbf{M}^n_4$ gates are

\begin{equation} \label{eq:M_1}
\begin{array}{rrrrr}
\mathbf{M}^1_4=\left[
\begin{array}{cccc}
0 & 0 & 0 &1 \\
1 & 0 & 0 & 0\\
0 & 1 & 0 & 0\\
0 & 0& 1 & 0
\end{array} \right],
\end{array}
\end{equation}

\begin{equation} \label{eq:M_2}
\begin{array}{rrrrr}
\mathbf{M}^2_4=\left[
\begin{array}{cccc}
0 & 0 & 1 &0 \\
0 & 0 & 0 & 1\\
1 & 0 & 0 & 0\\
0 & 1& 0 & 0
\end{array} \right],
\end{array}
\end{equation}

\begin{equation} \label{eq:M_3}
\begin{array}{rrrrr}
\mathbf{M}^3_4=\left[
\begin{array}{cccc}
0 & 1 & 0 &0 \\
0 & 0 & 1 & 0\\
0 & 0 & 0 & 1\\
1 & 0& 0 & 0
\end{array} \right].
\end{array}
\end{equation}

The qubit Pauli-$\mathbf{Z}$ operator can also be generalized into qudit space using the $d^{th}$ roots of unity. In the generalized Pauli group (Heisenberg-Weyl group), the number of $\mathbf{Z}$-operators depends on the system dimension, $d$. The transformation matrix becomes

\begin{equation}
    \mathbf{Z}^n_d= \begin{bmatrix*}[c]
        \omega^{n \cdot 0} & 0 & \dots & 0 \\
        0 & \omega^{n \cdot 1} & \dots & 0 \\
        \vdots & \vdots & \ddots & \vdots \\
        0 & 0 & \dots & \omega^{n \cdot (d-1)}
    \end{bmatrix*},
\end{equation}

\noindent where a phase shift on a particular qudit is produced via $\mathbf{Z}^n_d \Ket{k} = \omega^{nk} \Ket{k}$. The exponents $n$ range over $[1,d-1]$ as raising $\mathbf{Z}_d$ to the power of $d$ produces the identity gate. For a $d=4$ quantum system, we consider the following generalized $\mathbf{Z}$ operators:

\begin{equation} \label{eq:Z_1}
\begin{array}{rrrrr}
\mathbf{Z}^{1}_4=\left[
\begin{array}{cccc}
1 & 0 & 0 &0 \\
0 & i & 0 & 0\\
0 & 0 & -1 & 0\\
0 & 0& 0 & -i
\end{array} \right],
\end{array}
\end{equation}

\begin{equation} \label{eq:Z_2}
\begin{array}{rrrrr}
\mathbf{Z}^{2}_4=\left[
\begin{array}{cccc}
1 & 0 & 0 &0 \\
0 & -1 & 0 & 0\\
0 & 0 & 1 & 0\\
0 & 0& 0 & -1
\end{array} \right],
\end{array}
\end{equation}

\begin{equation} \label{eq:Z_2}
\begin{array}{rrrrr}
\mathbf{Z}^{3}_4=\left[
\begin{array}{cccc}
1 & 0 & 0 &0 \\
0 & -i & 0 & 0\\
0 & 0 & -1 & 0\\
0 & 0& 0 & i
\end{array} \right].
\end{array}
\end{equation}

Controlled variations of the $\mathbf{M}^n_d$ and $\mathbf{Z}^n_d$ operators are possible to create two-qudit gates. It is important to note that these controlled gates can be conditioned on the dimension-$d$ control qudit having the value of one of $d$ possible levels.

\subsection{Pauli-check based Qubit QED and QEM}

In $d=2$ quantum computing, the control qubit is not the only qubit influencing state change in a controlled-unitary operation, $CU$. Through the phenomenon known as phase kickback, the control qubit can also be affected by the target. In this case, a phase value (or eigenvalue) is kicked back from the target to the control if the target qubit's state is an eigenvector of the single-qubit unitary, $U$, embedded in the $CU$ gate. Phase kickback with $CU=CX$ is pictured in Fig.~\ref{fig:pk_2dim}. Phase kickback is crucial for Pauli-check sandwiching (PCS), forming the basis for its error mitigation technique.
\begin{figure}[htbp]
    \centering
    \begin{quantikz}
        \lstick{$\alpha\ket{0} + \beta\ket{1}$} & \ctrl{1} & \qw & \rstick{$\alpha\ket{0} - \beta\ket{1}$} \\
        \lstick{$\frac{1}{\sqrt{2}}(\ket{0} - \ket{1})$} & \targ{} & \qw & \rstick{$\frac{1}{\sqrt{2}}(\ket{0} - \ket{1})$}
    \end{quantikz}
    \caption{Phase kickback in qubits. Here, a phase of -1 gets kicked back onto the basis state $\ket{1}$ of the top qubit.}
    \label{fig:pk_2dim}
\end{figure}

PCS is an emerging technique to detect and mitigate errors in qubit-based quantum computing~\cite{gonzales2023quantum, qutracer, langfitt2024pauli, liu2025quantem}. PCS can also be leveraged for guiding quantum computer characterization and resource allocation~\cite{langfitt2024dynamic}. As seen in Fig.~\ref{fig:pcs-circuit}, PCS surrounds a payload circuit, $U$, with controlled Pauli operator checks. Errors on $U$ can be detected by measuring the ancilla qubit's value.
It is important that the relationship $R_{1}UL_{1} = U$ holds so the introduced checks do not affect the overall state of the payload circuit.

Observe that the relation $R_{1}UL_{1} = U$ implies $L_1 = U^{-1} R_1^{-1} U$, so the left check can be absorbed into $U$. In the equivalent circuit, $U$ commutes through the left check, leaving $R_1^{-1}$ and $R_1$ sandwiching only the error channel $\varepsilon$ following $U$. Consequently, the ancilla syndrome depends solely on the commutation relation between the error and the right checks, independent of $U$. This equivalent perspective, noted in~\cite{gonzales2023quantum}, allows us to view the Pauli checks as a transformation on the error channel itself, independent of the payload circuit. We make use of this interpretation in the proofs in Sec.~\ref{sec:general_pcs}.

The fact that all errors can be decomposed into the Pauli (or, more generally, Heisenberg-Weyl) basis is central to PCS. By exploiting this, we can construct the checks such that the underlying Pauli operators composing the error channel are detected via right checks that anti-commute with them (for instance, using a $Z$-type right check to detect $X$-type errors). The resulting phase kickback causes the ancilla to be measured in a state other than $\ket{0}$, indicating that an error is present. We then use this signal for QEM by discarding the error-corrupted shots, leading to increased fidelity in the final distribution.

\section{Generalizing Pauli-check Methods for Arbitrary Dimensions}
\label{sec:general_pcs}
PCS for any finite dimension $d$ depends on the existence of a generalized Pauli set within that dimension, and the $d^{th}$ roots of unity which form a cyclic group within the dimension. Because both structures scale linearly as the dimension increases, the order of operation for the PCS technique 
can be generalized to any finite dimension within the Hilbert space. Thus, the left and right checks for the generalized PCS method take the following forms:

\begin{equation}
    \label{eq:L1}
    L_X = \sum_{k=1}^{d-1} M_d^k \otimes \ket{k}\bra{k} +\mathbb{I}_d \otimes \ket{0}\bra{0}
\end{equation}
\begin{equation}
    \label{eq:L2}
    L_Z = \sum_{k=1}^{d-1} Z_d^{d-k} \otimes\ket{k}\bra{k} +\mathbb{I}_d \otimes \ket{0}\bra{0}
\end{equation}
where $L_X$ and $L_Z$ represent left checks for phase errors and basis permutation errors, respectively, and

\begin{equation}
    \label{eq:R1}
    R_X = \sum_{k=1}^{d-1} M_d^{d-k} \otimes \ket{k}\bra{k} +\mathbb{I}_d \otimes \ket{0}\bra{0}
\end{equation}
\begin{equation}
    \label{eq:R2}
    R_Z = \sum_{k=1}^{d-1} Z_d^k \otimes \ket{k}\bra{k}+\mathbb{I}_d \otimes \ket{0}\bra{0}
\end{equation}
where $R_X$ and $R_Z$ represent right checks for the respective phase and basis errors. The circuit diagram for PCS arbitrary qudit error detection is illustrated in Fig.~\ref{fig:arbitrary-pcs-any-d}. Under this scheme, the fidelity of the PCS mitigated circuit we obtain is 1, since all possible errors are detected using the checks. In theory, PCS can yield unity fidelity, assuming the checks used to implement the protocol are noiseless, as described in Ref.~\cite{gonzales2023quantum}, and we provide a proof for the qudit case in the following section. However, as circuit complexity increases for any dimension, it will become increasingly difficult to apply both error checks per qudit due to scalability. As a result, we limit our checks per qudit to only those that commute with the local gate operations on said qudits.

\begin{figure}[htb]
     \centering
         \includegraphics[width=0.9\linewidth,]{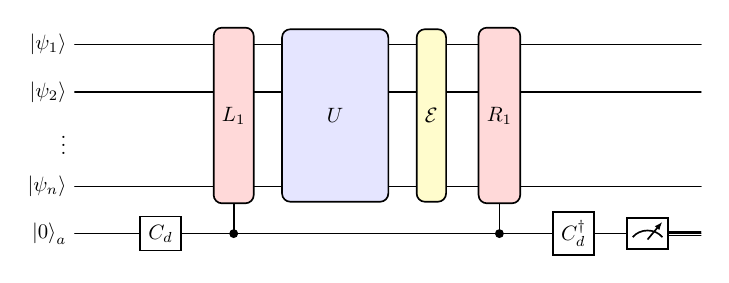}
        \caption{General PCS circuit layout. The red unitaries $L_1$ and $R_1$ represent the Pauli checks that sandwich the main payload circuit $U$. Measurement of the ancillas provides detection of the error channel $\varepsilon$. The Chrestenson gate $C_d$ is equivalent to the Hadamard gate when $d=2$.}
        \label{fig:pcs-circuit}
        
\end{figure}

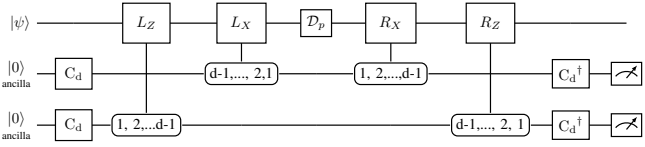
\begin{figure}[htbp]
    \centering
    \scalebox{0.57}{
    \begin{quantikz}
        \lstick{$\ket{\psi}$} & \qw & \gate[style={rectangle, draw, inner sep=5pt}]{L_Z} & \gate[style={rectangle, draw, inner sep=5pt}]{L_X} & \gate{\mathcal{D}_p} & \gate[style={rectangle, draw, inner sep=5pt}]{R_X} & \gate[style={rectangle, draw, inner sep=5pt}]{R_Z} & \qw & \qw \\
        \lstick{\shortstack{$\ket{0}$ \\ \scriptsize ancilla}} & \gate[style={rectangle, draw, inner sep=2pt}]{\mathrm{C_d}} & \qw & \gate[style={rectangle, draw, rounded corners, inner sep=-1pt}]{\text{d-1,..., 2,1}} \vqw{-1} & \qw & \gate[style={rectangle, draw, rounded corners, inner sep=-1pt}]{\text{1, 2,...,d-1}} \vqw{-1} & \qw & \gate{\mathrm{C_d}^\dagger} & \meter{} \\
        \lstick{\shortstack{$\ket{0}$ \\ \scriptsize ancilla}} & \gate[style={rectangle, draw, inner sep=2pt}]{\mathrm{C_d}} & \gate[style={rectangle, draw, rounded corners, inner sep=-1pt}]{\text{1, 2,...d-1}} \vqw{-2} & \qw & \qw & \qw & \gate[style={rectangle, draw, rounded corners, inner sep=-1pt}]{\text{d-1,..., 2, 1}} \vqw{-2} & \gate{\mathrm{C_d}^\dagger} & \meter{}
    \end{quantikz}
    }
    \caption{ Generalized PCS circuit for arbitrary error detection. Again, two ancillas are used as syndrome registers since two sets of checks are used.}
    \label{fig:arbitrary-pcs-any-d}
\end{figure}

\subsection{Detecting Errors in Arbitrary Dimensions}

\noindent
We consider data qudits of dimension~$d$.  For notational convenience, we write $X \equiv M^1_d$ and $Z \equiv Z^1_d$ for the generators of the shift and clock operators on a single qudit, so that $M^k_d = X^k$ and $Z^n_d = Z^n$.  The Heisenberg--Weyl (HW) operators on a single dimension-$d$ system are $\{X^a Z^b : a,b \in \Zd\}$, satisfying the commutation relation
\begin{equation}\label{eq:comm}
    Z^k X^a \;=\; \om^{ka}\, X^a Z^k, \qquad \om = e^{2\pi i / d}.
\end{equation}
An equivalent form that will be used for the $X$-check analysis is
\begin{equation}\label{eq:comm_dual}
    Z^b X^k \;=\; \om^{bk}\, X^k Z^b.
\end{equation}

\noindent
We use the checks $L_X$, $R_X$ and $L_Z$, $R_Z$ as defined in Eqs.~\eqref{eq:L1}--\eqref{eq:R2} of Section~\ref{sec:general_pcs}.  In the shorthand above, $L_X$ and $R_X$ are the $X$-checks (shift-type, detecting phase errors) and $L_Z$ and $R_Z$ are the $Z$-checks (clock-type, detecting basis permutation errors).  Together with the Chrestenson gate $\Cd$ on each ancilla, we refer to this as the \emph{generalized qudit PCS protocol}.  The checks satisfy
\begin{equation}\label{eq:check_identity}
    R_Z\, U\, L_Z = U, \qquad R_X\, U\, L_X = U,
\end{equation}
so that the introduced checks do not alter the action of the payload circuit~$U$.  In Propositions~\ref{prop:zcheck}--\ref{prop:xcheck} and Theorem~\ref{thm:unit}, we establish the error detection and unit fidelity results for a single data qudit, and Corollary~\ref{cor:multi} then extends the result to an $n$-qudit register by applying independent check pairs to each qudit.


\begin{proposition}[$Z$-check ancilla collapse]\label{prop:zcheck}
    Let $E = X^a Z^b$ be an arbitrary Heisenberg--Weyl error acting on the data qudit, and suppose the $Z$-check layer ($L_Z$, $R_Z$) of the generalized qudit PCS protocol is applied. Then the ancilla maps deterministically to $\ket{a \bmod d}$, identifying the shift index of the error. In particular, $E$ is undetected by the $Z$-checks if and only if $a \bmod d = 0$.
\end{proposition}

\begin{proof}
Let $\ket{\psi}_D$ denote the state of the data qudit and let the ancilla be initialized to $\ket{0}_A$.  After applying the Chrestenson gate $\Cd$ to the ancilla, we have the product state
\begin{equation}\label{eq:step1}
    \ket{\psi}_D \otimes \frac{1}{\sqrt{d}} \sum_{k=0}^{d-1} \ket{k}_A.
\end{equation}

\noindent The left check $L_Z$ in~\eqref{eq:L2} acts as $Z^{d-k}$ on the data qudit when the ancilla is in state $\ket{k}$, and as the identity when $k=0$.  The joint state becomes
\begin{equation}\label{eq:step2}
    \frac{1}{\sqrt{d}} \sum_{k=0}^{d-1} Z^{d-k} \ket{\psi}_D \otimes \ket{k}_A,
\end{equation}
where we adopt the convention $Z^{d-0} = Z^d = \Id$ so that the $k=0$ term is included in the sum.

\medskip
\noindent
The error $E = X^a Z^b$ acts on the data register:
\begin{equation}\label{eq:step3}
    \frac{1}{\sqrt{d}} \sum_{k=0}^{d-1} X^a Z^b\, Z^{d-k} \ket{\psi}_D \otimes \ket{k}_A.
\end{equation}

\noindent The right check $R_Z$ in~\eqref{eq:R2} applies $Z^k$ to the data qudit conditioned on ancilla state $\ket{k}$:
\begin{equation}\label{eq:step4}
    \frac{1}{\sqrt{d}} \sum_{k=0}^{d-1} Z^k\, X^a Z^b\, Z^{d-k} \ket{\psi}_D \otimes \ket{k}_A.
\end{equation}

\noindent Using the commutation relation~\eqref{eq:comm}, we rewrite the leading $Z^k X^a$ factor as $Z^k X^a = \om^{ka}\, X^a Z^k$.
Substituting into~\eqref{eq:step4} gives
\begin{equation}\label{eq:step5a}
    \frac{1}{\sqrt{d}} \sum_{k=0}^{d-1} \om^{ka}\, X^a Z^k Z^b Z^{d-k} \ket{\psi}_D \otimes \ket{k}_A.
\end{equation}
Since the $Z$ operators are mutually commuting, the exponents combine as
\begin{equation}\label{eq:Zcombine}
    Z^k Z^b Z^{d-k} = Z^{k + b + d - k} = Z^{b+d} = Z^b,
\end{equation}
where the last equality uses $Z^d = \Id$.  Hence~\eqref{eq:step5a} simplifies to
\begin{equation}\label{eq:step5b}
    \frac{1}{\sqrt{d}}\, X^a Z^b \ket{\psi}_D \otimes \sum_{k=0}^{d-1} \om^{ka} \ket{k}_A.
\end{equation}
Note that the data register has factored out of the sum, carrying the error $X^a Z^b$ uniformly across all branches.

\medskip
\noindent The inverse Chrestenson gate acts on the ancilla as $\Cd^\dagger \ket{k} = \frac{1}{\sqrt{d}} \sum_{m=0}^{d-1} \om^{-km} \ket{m}$.  Applying it to the ancilla state in~\eqref{eq:step5b} yields
\begin{equation}\label{eq:step6}
    X^a Z^b \ket{\psi}_D \otimes \frac{1}{d} \sum_{m=0}^{d-1} \left(\sum_{k=0}^{d-1} \om^{k(a-m)}\right) \ket{m}_A.
\end{equation}

\noindent The inner sum is evaluated by the orthogonality of the $d$th roots of unity:
\begin{equation}\label{eq:fourier}
    \frac{1}{d} \sum_{k=0}^{d-1} \om^{k(a-m)} = \delta_{a,m \bmod d}.
\end{equation}
Substituting~\eqref{eq:fourier} into~\eqref{eq:step6}, the ancilla collapses to a single basis state and the final joint state is
\begin{equation}\label{eq:final}
    X^a Z^b \ket{\psi}_D \otimes \ket{a \bmod d}_A.
\end{equation}
The ancilla readout therefore deterministically identifies the shift index~$a$ of the error.  Since $E$ yields $\ket{0}_A$ if and only if $a \equiv 0 \pmod{d}$, the $Z$-checks leave undetected precisely the pure phase errors $\{Z^b : b \in \Zd\}$.
\end{proof}


\begin{proposition}[$X$-check ancilla collapse]\label{prop:xcheck}
    Let $E = X^a Z^b$ be an arbitrary Heisenberg--Weyl error acting on the data qudit, and suppose the $X$-check layer ($L_X$, $R_X$) of the generalized qudit PCS protocol is applied. Then the ancilla maps deterministically to $\ket{b \bmod d}$, identifying the phase index of the error. In particular, $E$ is undetected by the $X$-checks if and only if $b \bmod d = 0$.
\end{proposition}

\begin{proof}
    The proof is analogous to that of Proposition~\ref{prop:zcheck}, with the roles of the shift and clock operators interchanged. The commutation relation~\eqref{eq:comm} is applied in the dual form~\eqref{eq:comm_dual}, $Z^b X^k = \om^{bk} X^k Z^b$, extracting the phase factor $\om^{bk}$ in place of $\om^{ka}$, and the shift operators combine as $X^{d-k} X^a X^k = X^a$ via $X^d = \Id$. The remainder of the argument proceeds identically, with Fourier orthogonality collapsing the ancilla to $\ket{b \bmod d}_A$. We provide the full step-by-step derivation in the Appendix.
\end{proof}


\begin{theorem}[Unit fidelity under joint post-selection]\label{thm:unit}
Let $\mathcal{E}$ be an arbitrary completely positive, trace-preserving (CPTP) error channel acting on a single data qudit of dimension~$d$, with Kraus representation $\mathcal{E}(\rho) = \sum_\ell K_\ell \,\rho\, K_\ell^\dagger$.  Under the generalized qudit PCS protocol with two check layers---the $Z$-check layer ($L_Z$, $R_Z$) and the $X$-check layer ($L_X$, $R_X$), each with a dedicated ancilla---postselecting both ancillas on $\ket{0}$ yields unit fidelity in the noiseless-check limit.
\end{theorem}

\begin{proof}
By the identity~\eqref{eq:check_identity}, the left checks can be absorbed into~$U$, so that the checks act solely as a transformation on the error channel, independent of the payload circuit (cf.~\cite{gonzales2023quantum}).  In the equivalent circuit, each check layer transforms the error by conjugation and post-selection.  As shown in the proof of Proposition~\ref{prop:zcheck}, the $Z$-check layer maps each error component $E$ to $\frac{1}{d}\sum_{k=0}^{d-1} Z^k\, E\, Z^{-k}$, and by the analogous argument in Proposition~\ref{prop:xcheck}, the $X$-check layer maps $E$ to $\frac{1}{d}\sum_{j=0}^{d-1} X^{-j}\, E\, X^{j}$.  Composing both layers, the postselected output state is
\begin{equation}\label{eq:thm_rho_ps}
    \rho_{\mathrm{ps}} = \frac{\sum_\ell K'_\ell\, (U\rho_0 U^\dagger)\, K'^{\,\dagger}_\ell}{\mathrm{Tr}\!\left[\sum_\ell K'_\ell\, (U\rho_0 U^\dagger)\, K'^{\,\dagger}_\ell\right]},
\end{equation}
where the transformed Kraus operators are
\begin{equation}\label{eq:transformed_kraus}
    K'_\ell = \frac{1}{d^2} \sum_{k=0}^{d-1} \sum_{j=0}^{d-1} X^{-j}\, Z^k\, K_\ell\, Z^{-k}\, X^{j}.
\end{equation}
Each Kraus operator admits a unique expansion in the HW basis,
\begin{equation}\label{eq:kraus_expand}
    K_\ell = \sum_{a,b \in \Zd} c_{ab}^{(\ell)}\, X^a Z^b.
\end{equation}
Substituting~\eqref{eq:kraus_expand} into~\eqref{eq:transformed_kraus} and applying the commutation relations~\eqref{eq:comm} and~\eqref{eq:comm_dual}, the conjugation acts on each HW component as
\begin{equation}\label{eq:conjugation}
    X^{-j}\, Z^k\, (X^a Z^b)\, Z^{-k}\, X^{j} = \om^{ka + bj}\, X^a Z^b.
\end{equation}
Hence the transformed Kraus operator becomes
\begin{equation}\label{eq:kraus_collapsed}
    K'_\ell = \sum_{a,b} c_{ab}^{(\ell)}\, X^a Z^b \cdot \underbrace{\frac{1}{d}\sum_{k=0}^{d-1} \om^{ka}}_{\displaystyle\delta_{a,0}} \cdot \underbrace{\frac{1}{d}\sum_{j=0}^{d-1} \om^{bj}}_{\displaystyle\delta_{b,0}} = c_{00}^{(\ell)}\, \Id.
\end{equation}
Since $K'_\ell \propto \Id$ for every~$\ell$, the proportionality constants cancel under normalization in~\eqref{eq:thm_rho_ps}, and the postselected state reduces to $\rho_{\mathrm{ps}} = U\rho_0\, U^\dagger$.  Therefore $F = \bra{\psi} \rho_{\mathrm{ps}} \ket{\psi} = 1$.
\end{proof}


\begin{corollary}[Multi-qudit unit fidelity]\label{cor:multi}
    Let $U$ be a unitary acting on $n$ data qudits, each of dimension~$d$, and let $\mathcal{E}$ be an arbitrary CPTP error channel acting on the data register.  If the generalized qudit PCS protocol is applied independently to each data qudit---with a dedicated $Z$-check ancilla and $X$-check ancilla per qudit, for a total of $2n$ ancillas---then postselecting all ancillas on $\ket{0}$ yields unit fidelity in the noiseless-check limit.
\end{corollary}

\begin{proof}
The HW operators on the $n$-qudit Hilbert space $(\mathbb{C}^d)^{\otimes n}$ are the tensor products $\{X^{a_1}Z^{b_1} \otimes \cdots \otimes X^{a_n}Z^{b_n} : a_i, b_i \in \Zd\}$, and these form an orthonormal basis for the space of $d^n \times d^n$ operators.  Each Kraus operator of the error channel can therefore be expanded as
\begin{equation}\label{eq:cor_kraus}
    K_\ell = \sum_{\mathbf{a}, \mathbf{b}} c_{\mathbf{a}\mathbf{b}}^{(\ell)}\; \bigotimes_{i=1}^{n} X^{a_i} Z^{b_i},
\end{equation}
where $\mathbf{a} = (a_1, \dots, a_n)$ and $\mathbf{b} = (b_1, \dots, b_n)$.

Since the checks on distinct qudits act on different tensor factors and are controlled by independent ancillas, the transformed Kraus operator under all $2n$ check layers factorizes as
\begin{equation}\label{eq:cor_transform}
    K'_\ell = \sum_{\mathbf{a},\mathbf{b}} c_{\mathbf{a}\mathbf{b}}^{(\ell)} \bigotimes_{i=1}^{n} X^{a_i} Z^{b_i} \;\cdot\; \prod_{i=1}^{n} \underbrace{\frac{1}{d}\sum_{k_i=0}^{d-1} \om^{k_i a_i}}_{\displaystyle\delta_{a_i,0}} \cdot \underbrace{\frac{1}{d}\sum_{j_i=0}^{d-1} \om^{j_i b_i}}_{\displaystyle\delta_{b_i,0}}.
\end{equation}
The Fourier sums enforce $a_i = b_i = 0$ for every $i = 1, \dots, n$.  Thus $K'_\ell = c_{\mathbf{0}\mathbf{0}}^{(\ell)}\, \Id^{\otimes n}$, and the result follows from Theorem~\ref{thm:unit} by the same normalization argument.
\end{proof}

\section{Ququart Pauli-check Methods}

Having established the generalized PCS framework for arbitrary dimension $d$ in Section~\ref{sec:general_pcs}, we now illustrate the protocol concretely in the ququart ($d = 4$) space. In higher-dimensional Hilbert spaces, the number of possible states increases, and each state has a specific controlled-gate implementation that enables phase kickback. For a superposition state, more than one controlled operation is needed to kick back a phase onto more than one basis state, as illustrated in Fig.~\ref{fig:pk_4dim}.

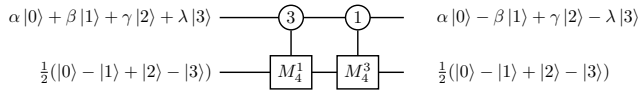
\begin{figure}[htbp]
    \centering
    \scalebox{0.67}{
        \begin{quantikz}
            \lstick{$\alpha \ket{0} + \beta \ket{1} + \gamma \ket{2} + \lambda \ket{3}$} & \qw & \gate[style={circle,draw,inner sep=-2pt}]{3} \vqw{1} & \gate[style={circle,draw,inner sep=-2pt}]{1} \vqw{1} & \qw & \rstick{$\alpha \ket{0} - \beta \ket{1} + \gamma \ket{2} - \lambda \ket{3}$} \\
            \lstick{$\frac{1}{2} ( \ket{0} - \ket{1} + \ket{2} - \ket{3})$} & \qw & \gate{M^1_4} & \gate{M^3_4} & \qw & \rstick{$\frac{1}{2} ( \ket{0} - \ket{1} + \ket{2} - \ket{3})$}
        \end{quantikz}
    }
    \caption{Phase kickback in dimension-4 Hilbert Space. Here, a phase of -1 gets kicked back onto the basis states of $\ket{1}$ and $\ket{3}$ of the top ququart.}
    \label{fig:pk_4dim}
\end{figure}

An important note about phase kickback is that it requires commutation relations between the Pauli/generalized Pauli group (Heisenberg-Weyl group) to send a phase from one qubit/qudit to another. In Fig.~\ref{fig:pk_2dim} and Fig.~\ref{fig:pk_4dim}, it can be observed that during the preparation of the state, the phase of the target qubit/qudit can be initialized using the gate $\mathbf{Z}$ for the two-dimensional circuit and the gates $\mathbf{Z}^{2}_{4}$ for the dimension-4 circuit (after the qubit/qudit is initialized to an equal and maximal superposition without phase). Because $\mathbf{Z}$ and $\mathbf{X}$ commute only up to a nontrivial phase, 
a phase appears on the top qubit/qudit.

\subsection{Detecting Ququart Basis Flips}

In $d=4$ Hilbert space, the generalized Pauli group has phase factors $1$, $i$, $-1$, and $-i$, which are the 4th roots of unity $\omega^k$ for $k \in \{0,1,2,3\}$ with $\omega = e^{2\pi i/4}$. These represent the different types of phase factors appearing from commutation relations in dimension 4. We assume that the target circuits in our experiments contain only Pauli \textbf{X} and \textbf{Z} gate operations. As such, when applying PCS in the ququart space to detect permutation or basis errors, the controlled $\mathbf{Z}$ operations are used to implement the left and right checks because they always commute with $\mathbf{Z}$ type operators in the target circuit and anti-commute with almost all other gates. This distinction from the number of anticommuting gates is relevant because the $\mathbf{Z^{2}_4}$ and $\mathbf{M^{2}_4}$  gates commute with each other. Specifically, the commutation relation $Z^n M^k = \omega^{nk} M^k Z^n$ yields a trivial phase factor $\omega^{nk} = 1$ whenever $nk \equiv 0 \pmod{d}$. For $n = k = 2$ in dimension $d = 4$, we have $\omega^{2 \cdot 2} = \omega^4 = 1$, so $Z^2_4$ alone cannot detect an $M^2_4$ error via phase kickback. This is precisely why the full sequence of all $d-1$ check operators is required in each check layer: while any individual $Z^n_4$ may commute with a particular $M^k_4$, the complete set $\{Z^1_4, Z^2_4, Z^3_4\}$ ensures that every non-identity shift error maps to a particular state in the Fourier basis, allowing the ancilla readout to uniquely identify the error (as stated in Propositions~\ref{prop:zcheck} and \ref{prop:xcheck}).

For the ququart circuits, the left and right check sequences are implemented using these operations:
\begin{equation}
    L_Z = Z_{4}^1 \otimes \ket{3}\bra{3} + Z_{4}^2 \otimes \ket{2}\bra{2}+Z_{4}^3 \otimes \ket{1}\bra{1}+\mathbb{I} \otimes \ket{0}\bra{0}
\end{equation}
\begin{equation}
    R_Z = Z_{4}^1 \otimes \ket{1}\bra{1} + Z_{4}^2 \otimes \ket{2}\bra{2}+Z_{4}^3 \otimes \ket{3}\bra{3}+\mathbb{I} \otimes \ket{0}\bra{0}
\end{equation}

\noindent Note that these are precisely the generalized checks of Eqs.~\eqref{eq:L2} and~\eqref{eq:R2} evaluated at $d = 4$. From Fig.~\ref{fig:pcs-ququart-basis}, it can be seen that all the Z gates in the ququart space are necessary to detect any arbitrary ququart basis error despite $\mathbf{Z^{2}_4}$ and $\mathbf{M^{2}_4}$ commuting with each other. This is because, for the ancilla readout to reflect the error type detected, all the gate operations possible in the Hilbert space need to be implemented in a particular sequence such that the phase kicked back onto the ancilla  causes the right type of interference between the superposed states of the ancilla ququart to yield the state which reflects the error within the circuit. If not all sets of checks are used, the readout fails to distinguish the type of error that has been detected.
The PCS protocol for a single-layer detection of all three basis errors in the ququart space is therefore represented in Fig.~\ref{fig:pcs-ququart-basis}.

\begin{figure}[htbp]
    \centering
    \scalebox{0.6}
    {\begin{quantikz}
        \ket{\psi} & \qw & \qw & \gate{Z^1_4} & \gate{Z^2_4} & \gate{Z^3_4} & \gate{M_{error}} & \gate{Z^1_4} & \gate{Z^2_4}  & \gate{Z^3_4} & \qw & \qw \\
        {\shortstack{$\ket{0}$ \\ \scriptsize ancilla}} & \qw & \gate[style= {rectangle, draw, inner sep=1pt}]{\mathrm{C_4}} & \gate[style={circle,draw,inner sep=-2pt}]{3} \vqw{-1} & \gate[style={circle,draw,inner sep=-2pt}]{2} \vqw{-1} & \gate[style={circle,draw,inner sep=-2pt}]{1} \vqw{-1} & \qw & \gate[style={circle,draw,inner sep=-2pt}]{1} \vqw{-1} & \gate[style={circle,draw,inner sep=-2pt}]{2} \vqw{-1} & \gate[style={circle,draw,inner sep=-2pt}]{3} \vqw{-1} & \gate{\mathrm{C_4}^\dagger} & \meter{}
    \end{quantikz}}
    \caption{PCS Ququart circuit for Basis error detection}
    \label{fig:pcs-ququart-basis}
\end{figure}
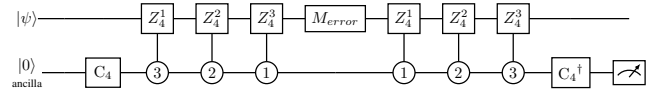
Here, the Hadamard gates used in the two-dimensional protocol are replaced with the Chrestenson operator and its conjugate transpose on the left and right sides of the checks respectively. For basis error detection, the right check induces a phase kickback onto the ancilla based on the Z-eigenvalue of the target state as a result of the error. This is because, the error is in the form of a basis change, and the phase factor of the ququart initialized in the $\ket{0}$ state is not changed until its state is changed by the error. This maps the target ququart into a different eigenspace with a phase factor not equal to 1, which is then kicked back onto the ancilla after encountering the right check.

\subsection{Detecting Ququart Phase Flips}
When detecting phase errors, the same stabilizer logic also applies. The only difference is that controlled-X operations are used to construct the checks. Mathematically, the checks for arbitrary phase error detection are given as:
\begin{equation}
    L_X = M_{4}^1 \otimes \ket{1}\bra{1} + M_{4}^2 \otimes \ket{2}\bra{2}+M_{4}^3 \otimes \ket{3}\bra{3}+\mathbb{I}_4 \otimes \ket{0}\bra{0}
\end{equation}
\begin{equation}
    R_X = M_{4}^1 \otimes \ket{3}\bra{3} + M_{4}^2 \otimes \ket{2}\bra{2}+M_{4}^3 \otimes \ket{1}\bra{1}+\mathbb{I}_4 \otimes \ket{0}\bra{0}
\end{equation}

\noindent Also, because the state of the target ququart is altered by the left X checks, the Z-type error maps the error onto the ancilla earlier than in the basis error detection. Thereafter, a similar evolution takes place wherein the target qudit is returned to its expected state after unitary transformation, and the error from the target circuit is read out from the ancilla as a permutation change. The circuit scheme for phase error detection in dimension 4 is illustrated in Fig.~\ref{fig:pcs-ququart-phase}.

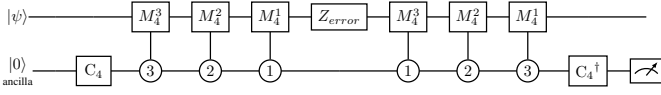
\begin{figure}[t]
    \centering
    \scalebox{0.6}
    {\begin{quantikz}
        \ket{\psi} & \qw & \qw & \gate{M^3_4} & \gate{M^2_4} & \gate{M^1_4} & \gate{Z_{error}} & \gate{M^3_4} & \gate{M^2_4}  & \gate{M^1_4} & \qw & \qw \\
        {\shortstack{$\ket{0}$ \\ \scriptsize ancilla}} & \qw & \gate[style= {rectangle, draw, inner sep=1pt}]{\mathrm{C_4}} & \gate[style={circle,draw,inner sep=-2pt}]{3} \vqw{-1} & \gate[style={circle,draw,inner sep=-2pt}]{2} \vqw{-1} & \gate[style={circle,draw,inner sep=-2pt}]{1} \vqw{-1} & \qw & \gate[style={circle,draw,inner sep=-2pt}]{1} \vqw{-1} & \gate[style={circle,draw,inner sep=-2pt}]{2} \vqw{-1} & \gate[style={circle,draw,inner sep=-2pt}]{3} \vqw{-1} & \gate{\mathrm{C_4}^\dagger} & \meter{}
    \end{quantikz}}
    \caption{PCS Ququart circuit for Phase error detection}
    \label{fig:pcs-ququart-phase}
\end{figure}

\subsection{Arbitrary Ququart Error Detection}

For circuits with both phase flip and basis change errors present, a combination of the two error detection techniques is employed. In this case, the basis check is sandwiched between the phase checks. The checks are implemented in this order to prevent the target qudits from undergoing phase changes before encountering the error. For arbitrary error detection in ququart space, we utilize the generalized error channel for higher-dimensional Hilbert spaces to simulate the error on the target unitary operation.
We also describe the PCS scheme for the checks as follows.
\begin{equation}
    L_X = M_{4}^1 \otimes \ket{1}\bra{1} + M_{4}^2 \otimes \ket{2}\bra{2}+M_{4}^3 \otimes \ket{3}\bra{3}+\mathbb{I}_4 \otimes \ket{0}\bra{0}
\end{equation}

\begin{equation}
    L_Z = Z_{4}^1 \otimes \ket{3}\bra{3} + Z_{4}^2 \otimes \ket{2}\bra{2}+Z_{4}^3 \otimes \ket{1}\bra{1}+\mathbb{I}_4 \otimes \ket{0}\bra{0}
\end{equation}
\begin{equation}
    R_X = M_{4}^1 \otimes \ket{3}\bra{3} + M_{4}^2 \otimes \ket{2}\bra{2}+M_{4}^3 \otimes \ket{1}\bra{1}+\mathbb{I}_4 \otimes \ket{0}\bra{0}
\end{equation}

\begin{equation}
    R_Z = Z_{4}^1 \otimes \ket{1}\bra{1} + Z_{4}^2 \otimes \ket{2}\bra{2}+Z_{4}^3 \otimes \ket{3}\bra{3}+\mathbb{I}_4 \otimes \ket{0}\bra{0}
\end{equation}

\noindent where $L_X$ and $L_Z$ represent left checks for phase errors and
basis permutation errors, and $R_X$ and $R_Z$ represent right checks for the respective
phase and basis errors. The PCS scheme for arbitrary ququart error therefore becomes what is seen in Fig.~\ref{fig:arbitrary-ququart}.

\begin{figure}[htbp]
    \centering
    \scalebox{0.6}{
    \begin{quantikz}
        \lstick{$\ket{\psi}$} & \qw & \gate[style={rectangle, draw, inner sep=5pt}]{L_Z} & \gate[style={rectangle, draw, inner sep=5pt}]{L_X} & \gate{\mathcal{D}_p} & \gate[style={rectangle, draw, inner sep=5pt}]{R_X} & \gate[style={rectangle, draw, inner sep=5pt}]{R_Z} & \qw & \qw \\
        \lstick{\shortstack{$\ket{0}$ \\ \scriptsize ancilla}} & \gate[style={rectangle, draw, inner sep=2pt}]{\mathrm{C_4}} & \qw & \gate[style={rectangle, draw, rounded corners, inner sep=-1pt}]{\text{3, 2, 1}} \vqw{-1} & \qw & \gate[style={rectangle, draw, rounded corners, inner sep=-1pt}]{\text{1, 2, 3}} \vqw{-1} & \qw & \gate{\mathrm{C_4}^\dagger} & \meter{} \\
        \lstick{\shortstack{$\ket{0}$ \\ \scriptsize ancilla}} & \gate[style={rectangle, draw, inner sep=2pt}]{\mathrm{C_4}} & \gate[style={rectangle, draw, rounded corners, inner sep=-1pt}]{\text{1, 2, 3}} \vqw{-2} & \qw & \qw & \qw & \gate[style={rectangle, draw, rounded corners, inner sep=-1pt}]{\text{3, 2, 1}} \vqw{-2} & \gate{\mathrm{C_4}^\dagger} & \meter{}
    \end{quantikz}
    }
    \caption{PCS Ququart circuit for arbitrary error detection. Here, two ancillas are used as syndrome registers since two sets of checks are used.}
    \label{fig:arbitrary-ququart}
\end{figure}
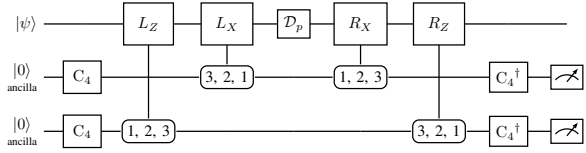

\section{Evaluation of Qudit PCS}


\subsection{Simulation Setup}
 
Our experimental analysis evaluates the fidelity and post-selection acceptance rate of PCS-mitigated circuits across varying Hilbert-space dimensions. All simulations are implemented in Cirq~\cite{cirq} using density matrix evolution, which allows exact computation of output state fidelity without sampling noise. We model errors using the $d$-dimensional depolarizing channel,
\begin{equation}
    \mathcal{E}(\rho) = (1-p)\,\rho + \frac{p}{d^2 - 1}\sum_{(a,b)\neq(0,0)} X^a Z^b\,\rho\,(X^a Z^b)^\dagger,
    \label{eq:depolarizing}
\end{equation}
where $p$ is the error rate and the sum runs over all $d^2 - 1$ non-identity Heisenberg-Weyl operators. A single depolarizing channel is applied after each one-qudit gate, while two independent channels (one per qudit) are applied following each two-qudit gate. We assume that errors occur only within the target circuit. The check gates are treated as noiseless so that we can isolate the detection performance of the PCS protocol itself.
 
We evaluate PCS on two state preparation circuits. The first, which we denote Mod-2, applies a single $M^2_d$ gate to a qudit initialized in $|0\rangle$, preparing the state $|2 \bmod d\rangle$.
The second circuit, denoted GHZ-2, prepares a two-qudit generalized GHZ state $\frac{1}{\sqrt{d}}\sum_{k=0}^{d-1}|k,k\rangle$ using a Chrestenson gate followed by a controlled-$X$ operation. This circuit contains both shift and phase-type operations, requiring two layers of checks with two ancillas for full error detection. It is important to note that there was only one check per layer for both experiments, so the circuits were not fully protected from all possible types of error. However, our results found that even lightweight PCS protection (i.e. checks where a single-qudit basis state controlled a single-qudit X or Z operation on the target) are effective against the unmitigated baseline. The single-qudit checks were selected using commutation relations from the payload layer.
 
In both cases, fidelity is computed as $F = \langle \psi_{\text{ideal}} | \rho_{\text{out}} | \psi_{\text{ideal}} \rangle$, where $\rho_{\text{out}}$ is the output density matrix after post-selection on the $|0\rangle$ ancilla outcome(s). The post-selection acceptance rate is the probability of obtaining $|0\rangle$ on all ancilla measurements.

\begin{figure}[!htbp]
    \centering
    \includegraphics[width=0.48\textwidth]{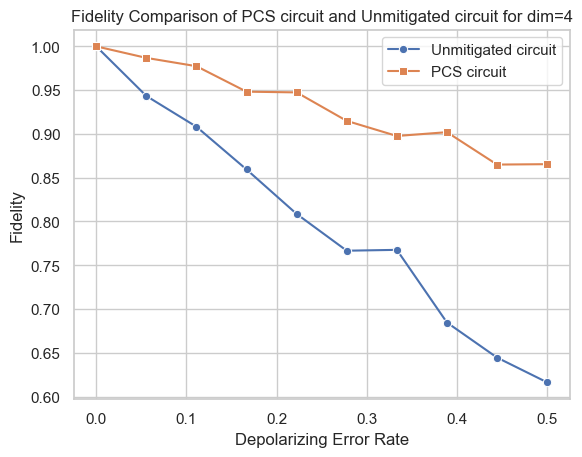}
    \caption{Fidelity comparison of PCS mitigated vs. unmitigated Modulo Addition 2 circuit in dimension 4.}
    \label{fig:PCS-mod2-dim4}
\end{figure}

\begin{figure}[!htbp]
    \centering
    \includegraphics[width=\columnwidth]{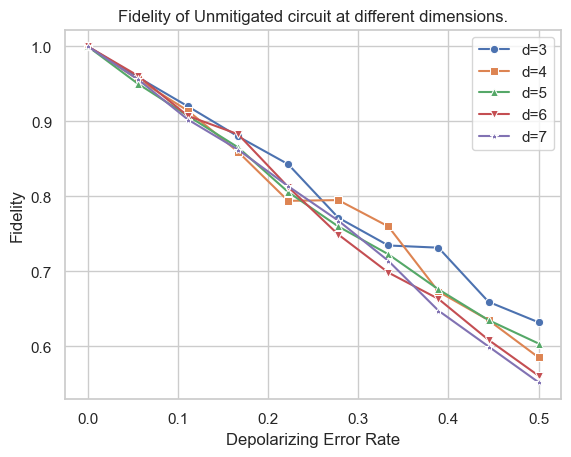}
    \vspace{0.2cm}
    \includegraphics[width=\columnwidth]{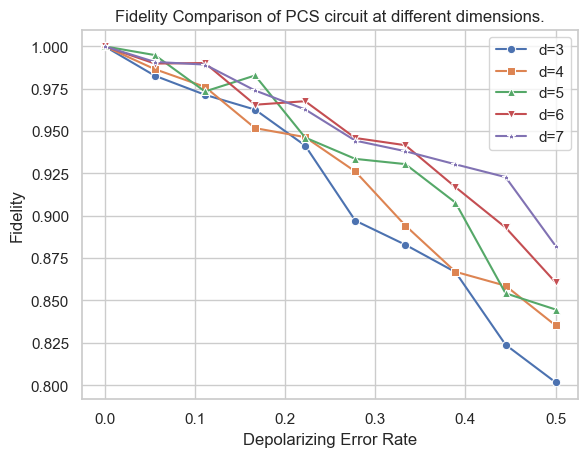}
    \caption{Fidelity of Mod-2 circuits across dimensions $d=3$ 
    through $d=7$. (Top) Unmitigated. (Bottom) PCS mitigated.}
    \label{fig:mod2-dims-comparison}
\end{figure}

\begin{figure}[!htbp]
    \centering
    \includegraphics[width=\columnwidth]{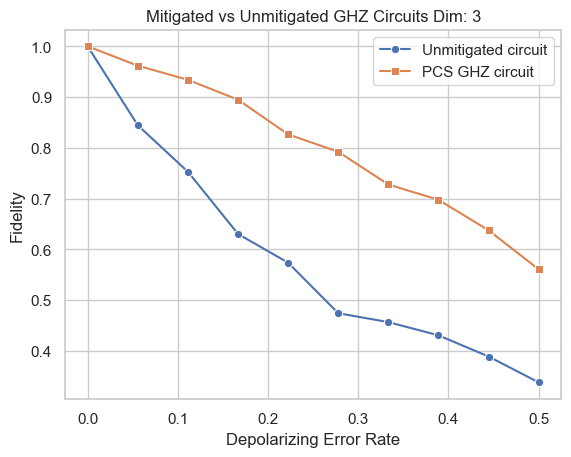}\\[0.3cm]
    \includegraphics[width=\columnwidth]{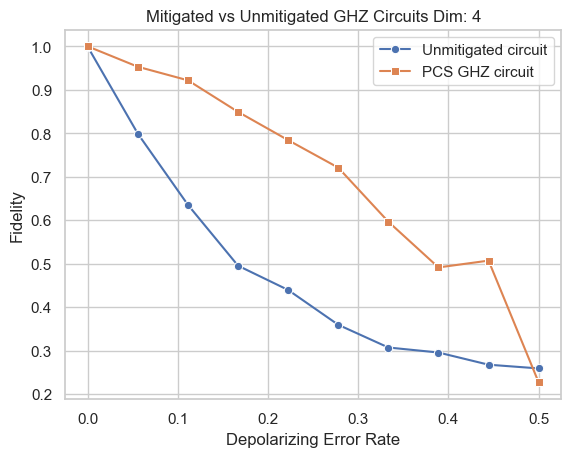}\\[0.3cm]
    \includegraphics[width=\columnwidth]{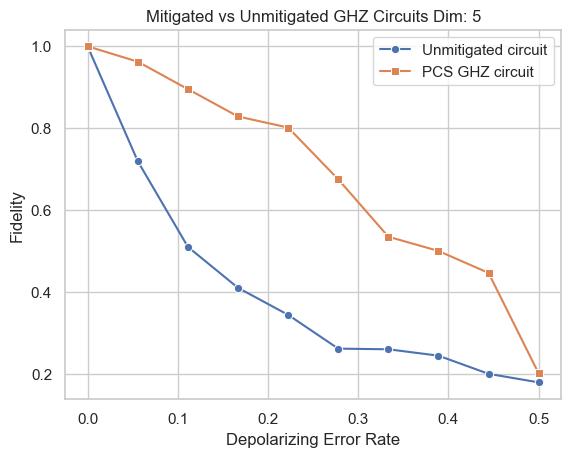}
    \caption{Fidelity comparison for PCS mitigated vs.\ unmitigated GHZ-2 circuits for dimensions 3, 4, and 5 (top to bottom).}
    \label{fig:ghz_dim_3_4_5}
\end{figure}

\subsection{Mod-2 Circuit Results}
We first compare the fidelity of the Mod-2 circuit mitigated with PCS against its unmitigated variant for the dimension-4 Hilbert space. In Fig. \ref{fig:PCS-mod2-dim4} we see PCS achieves fidelities between 0.9 and 1.0 at low error rates and maintains fidelity readings above 0.85 for higher error rates ($0.2 \le p \le 0.5$). For subsequent readings across dimensions $d=3$ to $7$, Fig.~\ref{fig:mod2-dims-comparison} shows PCS improvements relative to the unmitigated scenario, with higher dimensions showing improvements in fidelity. More specifically, we see that PCS is able to maintain fidelities above $0.80$ across all dimensions up to a depolarizing error rate of $0.5$.

Also observe that the relative fidelity improvements increase as the dimension increases. This dimensional scaling can be understood from the structure of the depolarizing channel in Eq.~\eqref{eq:depolarizing}. As $d$ increases, the total error probability $p$ is distributed across $d^2 - 1$ non-identity HW operators, so the weight of each individual error component decreases as $p/(d^2 - 1)$. Since PCS detects each non-identity component independently, the fraction of shots that pass post-selection while still carrying an undetected error shrinks with increasing dimension. Quantitatively, for the Mod-2 circuit at $p = 0.3$, PCS provides a fidelity improvement of approximately 20\% at $d = 3$, increasing to roughly 32\% at $d = 7$.



\subsection{GHZ-2 Circuit Results}



In contrast to Mod-2 circuits, we observe different results in the GHZ-2 circuit. Comparing the mitigated and unmitigated GHZ-2 circuit for dimensions 3, 4, and 5 (see Fig.~\ref{fig:ghz_dim_3_4_5}), we observe that fidelity results improve in the lower-noise regime (for $0<p<0.2$), but fidelity gains begin to deteriorate once $p$ exceeds $0.2$. 



When we consider only mitigated fidelity results across dimensions (see Fig.~\ref{fig:pcs-ghz-dims}), we observe a trend similar to that in Fig.~\ref{fig:ghz_dim_3_4_5}. Moreover, it can be observed that lower-dimensional circuits demonstrate stronger mitigation improvements relative to higher-dimensional ones. Although this may indicate that PCS performance decreases in higher dimensions, it does not show the robustness of the error mitigation technique in lower error regimes. This behavior is well demonstrated in Fig.~\ref{fig:pcs-p(0.025)-dims}, where the fidelity results remain highly stable across dimensions.



\begin{figure}[t]
    \centering
    \includegraphics[width=0.48\textwidth]{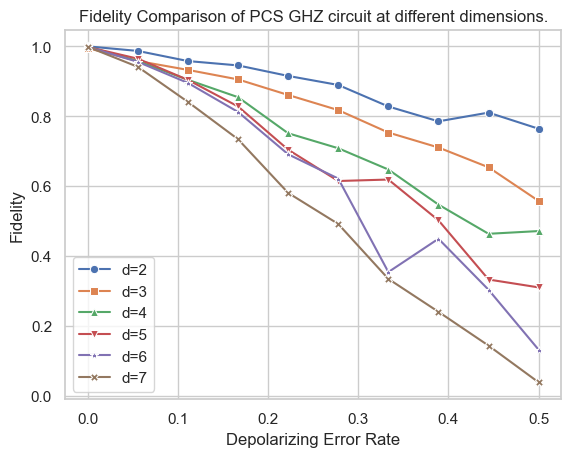}
    \caption{Comparison of fidelity performance of PCS mitigated GHZ-2 circuits across multiple dimensional Hilbert spaces from dim=2 to dim=7.}
    \label{fig:pcs-ghz-dims}
\end{figure}
\FloatBarrier

\begin{figure}[h]
    \centering
    \includegraphics[width=0.48\textwidth]{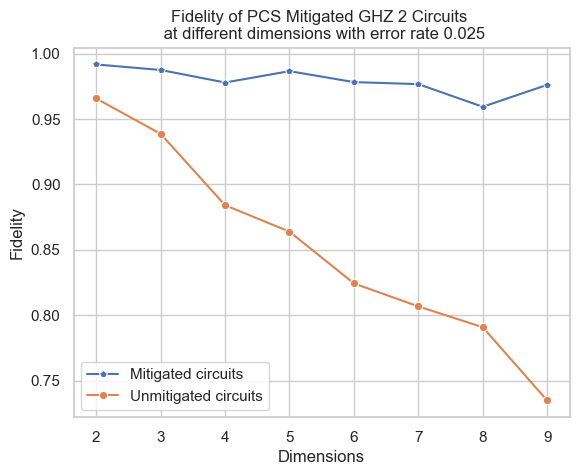}
    \caption{Fidelity of PCS mitigated vs unmitigated GHZ-2 circuit across multiple dimensional Hilbert spaces at error rate of 0.025. Fidelity values of mitigated are kept above roughly 97.5\% in a seemingly constant progression while unmitigated circuit decreases across dimensions 2 through to 9.} 
    \label{fig:pcs-p(0.025)-dims}
\end{figure}
\FloatBarrier

\subsection{Post-selection Overhead}

\noindent
It is important to note that PCS introduces an inherent post-selection tradeoff, which increases with increasing Hilbert-space dimension. As shown in Fig. ~\ref{fig:post-selection-measurements}, the Mod-2 circuit exhibits a seemingly linear decline in post-selection acceptance rates, while the GHZ-2 circuit shows an exponential decline that steepens as dimension increases.

This difference follows from the circuit structure: the Mod-2 circuit contains a single gate and requires only one ancilla, so the acceptance rate is governed by a single depolarizing channel. The GHZ-2 circuit, by contrast, contains both single- and two-qudit gates with two ancillas, meaning errors at each gate location independently contribute to the probability of failing post-selection. As dimension grows, the number of non-identity error components increases as $d^2 - 1$ per qudit, which compounds across the multiple error sites in the GHZ-2 circuit.

    
    
    



\begin{figure}[htbp]
    \centering
    
    \includegraphics[width=0.5\textwidth]{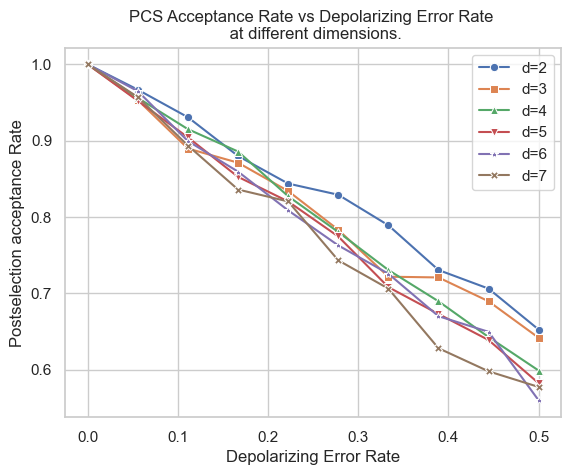}
    \vspace{0.3cm}
    
    \includegraphics[width=0.5\textwidth]{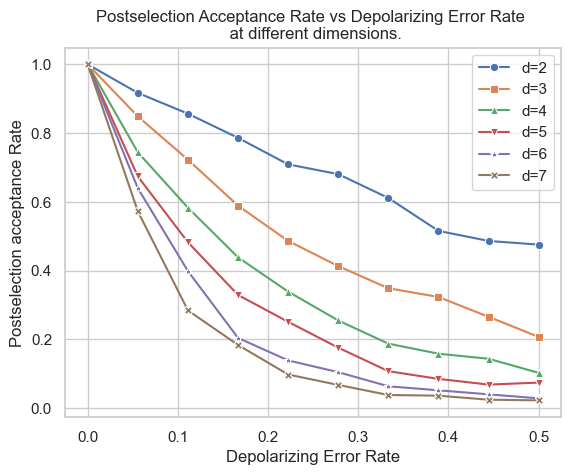}
    
    \caption{(Top) Post-selection acceptance rate of Mod-2 circuit. (Bottom) Post-selection acceptance rate of GHZ-2 circuit.}
    \label{fig:post-selection-measurements}
\end{figure}
\FloatBarrier

\section{Future Work}
Our analysis of Mod-2 circuits shows that generalized PCS becomes increasingly effective in higher-dimensional Hilbert spaces, particularly for low-depth circuits. Therefore, one potential area of investigation is to apply the Generalized PCS protocol to complex, lower-dimensional circuits compressed into higher dimensions to detect errors that arise during compression. This will provide further insight into the scalability of qudit PCS in such circuits in terms of improving the fidelity of complex circuits.

Another future line of work will be balancing tradeoffs associated with circuit depth and the mitigation performance of our working PCS protocol. Currently, our model detects errors by implementing the complete Pauli HW gate set. Reducing the number of gates per check while maintaining the error detection capabilities of generalized PCS will enable higher fidelity measurements at a lower circuit depth, especially when considering the possibility of the checks themselves being susceptible to errors. One area to look at may be deriving a structured subset of operators capable of detecting arbitrary phase and permutation errors. However, such a reduction may be dependent on the dimension and may not be able to capture all error types in some higher-dimensional systems.

    
    
    

    
    
    

Since the checks were assumed to be noiseless, one key development will be to investigate the performance of generalized PCS under real hardware constraints, where the checks are also prone to errors. This is currently challenging because most quantum computers available now are binary-level-based systems and do not offer capabilities for higher-dimensional computation.

Moreover, the gates used in this experiment are based on the theoretical Heisenberg-Weyl group, which may not efficiently map onto hardware-native gate sets. One area of future work, consequently, will be exploring state-of-the-art universal qudit gate sets that are compatible with current or near-term quantum hardware systems.

\section{Conclusion}
We successfully obtained a generalized qudit Pauli gate set for arbitrary dimensions, which was used to generalize the Pauli Check Sandwiching (PCS) protocol. 
Using this framework, we construct generalized left and right Pauli checks for arbitrary dimensions and illustrate the protocol concretely in the ququart ($d = 4$) space. In Propositions 1 and 2, we prove that the generalized PCS detects every error that does not commute with our checks by showing that every HW error on the payload circuit deterministically maps to a state that encodes the type of error, whereby M-type errors collapse to their corresponding error type under Z-checks and Z-type errors collapse into their respective error type under M-type checks. Since our checks are able to detect any Pauli-type error, by extension of Theorem 1 and under the assumption that our checks are noiseless, we prove that by applying the right number of layers of Pauli checks to detect M-type and Z-type errors, our PCS protocol can achieve unit fidelity for an arbitrary target. 

The performance of the generalized PCS protocol is tested using two types of circuits - a simple Mod-2 circuit and a GHZ-2 circuit - across multiple dimensions, comparing the fidelities of PCS-mitigated and unmitigated circuits. We simulate errors within our circuits using a dimension-dependent depolarizing error channel. For the simple Mod-2 circuit, we achieved fidelity gains ranging from roughly 20\% for $d=3$ to 32\% for $ d=7$. Our results show improvement in overall fidelity with increasing dimension, demonstrating that at low circuit depths, our PCS model improves the fidelity of circuits in higher dimensions.

For GHZ-2 simulations across dimensions, lower dimensions perform better than higher dimensions in terms of fidelity results after PCS error mitigation. However, at error rates below 0.05, PCS is observed to be robust across multiple dimensions, providing fidelity readouts above 90\%.

Overall, these results demonstrate that PCS is particularly effective when applied to compressed, low-depth circuits, where it can significantly improve fidelity and provide more robust computational outcomes, especially as system dimension increases.

\appendix[Proof of Proposition~2]\label{app:xcheck}

We provide the complete derivation for the $X$-check case.  Let $\ket{\psi}_D$ denote the state of the data qudit and let the ancilla be initialized to $\ket{0}_A$.  After the Chrestenson gate, left $X$-check $L_X$, error $E = X^a Z^b$, and right $X$-check $R_X$, the joint state is
\begin{equation}\label{eq:app_after_checks}
    \frac{1}{\sqrt{d}} \sum_{k=0}^{d-1} X^{d-k}\, X^a Z^b\, X^k \ket{\psi}_D \otimes \ket{k}_A,
\end{equation}
where we have used $X^0 = \Id$ to include the $k=0$ term in the sum.  Combining the leading shift operators gives $X^{d-k} X^a = X^{a+d-k}$.  We then commute $Z^b$ past $X^k$ using~\eqref{eq:comm_dual}, $Z^b X^k = \om^{bk} X^k Z^b$, and collect the shift exponents via $X^{a+d-k} X^k = X^{a+d} = X^a$ (since $X^d = \Id$).  The state simplifies to
\begin{equation}\label{eq:app_simplified}
    \frac{1}{\sqrt{d}}\, X^a Z^b \ket{\psi}_D \otimes \sum_{k=0}^{d-1} \om^{bk} \ket{k}_A.
\end{equation}
Note that the data register has factored out of the sum.  Applying the inverse Chrestenson gate $\Cd^\dagger$ to the ancilla, with $\Cd^\dagger \ket{k} = \frac{1}{\sqrt{d}} \sum_{m=0}^{d-1} \om^{-km} \ket{m}$, yields
\begin{equation}\label{eq:app_fourier_sum}
    X^a Z^b \ket{\psi}_D \otimes \frac{1}{d} \sum_{m=0}^{d-1} \left( \sum_{k=0}^{d-1} \om^{k(b-m)} \right) \ket{m}_A.
\end{equation}
By the orthogonality of the $d$th roots of unity,
\begin{equation}\label{eq:app_fourier}
    \frac{1}{d} \sum_{k=0}^{d-1} \om^{k(b-m)} = \delta_{b,m \bmod d},
\end{equation}
and the final joint state collapses to
\begin{equation}\label{eq:app_final}
    X^a Z^b \ket{\psi}_D \otimes \ket{b \bmod d}_A.
\end{equation}
The ancilla readout deterministically identifies the phase index~$b$ of the error.  Since $E$ yields $\ket{0}_A$ if and only if $b \equiv 0 \pmod{d}$, the $X$-checks leave undetected precisely the pure shift errors $\{X^a : a \in \Zd\}$. \qed

\bibliographystyle{plain}
\bibliography{refs}

\end{document}